%% file: paper.tex
\documentclass[a4paper,onecolumn,11pt,unpublished]{quantumarticle}
\pdfoutput=1
\input{settings}

\begin{document}

    \title{Quantum adders: on the structural link between the ripple-carry and carry-lookahead techniques}
    \author[1]{Maxime Remaud}
    \orcid{0009-0008-1597-3661}
    \affil[1]{Eviden Quantum Lab, Les Clayes-sous-Bois, France}
    
    \maketitle
    
    \begin{abstract}
        This paper is motivated by two key observations. 
        First, Toffoli ladders can be implemented in three distinct ways: with linear or polylogarithmic depth using no ancilla, or with logarithmic depth using ancilla qubits.
        Second, two fundamental structural approaches to designing addition algorithms can be identified in several well-known quantum adders. At their core is the Toffoli ladder, and both provide a clear and simple connection between ripple-carry and carry-lookahead adder designs.
        Combining these two structures with the three Toffoli ladder implementations yields six quantum adders: four are well-known and two novel. Notably, one of the novel designs is a carry-lookahead adder that outperforms previous approaches.
    \end{abstract}
    
    \input{0-introduction}
    \input{1-preliminaries}
    \input{2-ladders}
    \input{3-CLA}
    \input{4-conclusion}
    
    \bibliographystyle{splncs04}
    \bibliography{paper}

\end{document}

%% file: settings.tex
\usepackage[utf8]{inputenc}
\usepackage[english]{babel}
\usepackage[T1]{fontenc}
\usepackage{amsmath,amsthm,amssymb}
\usepackage{hyperref}
\usepackage[numbers]{natbib}
\usepackage{physics}
\usepackage{tikz,algorithm}
\usetikzlibrary{quantikz2}
\usepackage{circuitikz}
\usepackage[noend]{algpseudocode}
\usepackage{mathtools}

\hypersetup{
  colorlinks=true,
  linkcolor=red,        
  citecolor=red,         
  urlcolor=red       
}

\newtheorem{theorem}{Theorem}
\newtheorem{lemma}{Lemma}
\newtheorem{definition}{Definition}
\newtheorem{remark}{Remark}

\newcommand{\Add}{\mathsf{Add}}
\newcommand{\NOT}{\mathsf{NOT}}
\newcommand{\CARRY}{\mathsf{CARRY}}
\newcommand{\CNOT}{\mathsf{CNOT}}
\newcommand{\CCNOT}{\mathsf{CCNOT}}

\newcommand{\Xg}{\mathsf{X}}
\newcommand{\ladder}{\textsc{Ladder}}
\newcommand{\Lg}{\mathsf{L}}

\newcommand{\ie}{i.e.}
\newcommand{\etal}{et \; al.}
\newcommand*{\eqdef}{\stackrel{\text{def}}{=}}
\newcommand{\floor}[1]{\left\lfloor #1 \right\rfloor}

\newcommand{\ints}[1]{[\![ #1 ]\!]}
\newcommand{\Rbra}[1]{\left( #1 \right)}

\newcommand{\bigO}[1]{O\Rbra{#1}}
\newcommand{\hw}[1]{\omega\Rbra{#1}}

%% file: 0-introduction.tex
\section{Introduction}

Efficient arithmetic operations lie at the heart of both classical and quantum computing, with addition being one of the most fundamental. As quantum computing continues to mature, the design and optimization of quantum arithmetic circuits, particularly quantum adders, plays a critical role in enabling more complex algorithms such as cryptanalytic algorithms \cite{Kup05,HJN20}, quantum machine learning \cite{SCC24}, or even quantum chemistry \cite{CWM12}. Over the past three decades, various quantum adder architectures have been proposed, each with different trade-offs in terms of circuit depth, ancilla usage, gate count, and error resilience.

Three major families of quantum adders have emerged: quantum ripple-carry, quantum carry-lookahead, and QFT-based. The latter leverages the quantum Fourier transform, central to many quantum algorithms, to perform addition in the frequency domain. Introduced by Draper \cite{Dra02}, QFT-based addition requires higher gate precision and is more susceptible to errors introduced by phase rotations \cite{HRS17}. This method differs from those discussed in this paper, as it uses Hadamard gates and controlled rotations; hereafter, we will focus on adders using classical logic only.

The second family is that of quantum ripple-carry addition, which was first introduced by Vedral $\etal$ \cite{VBE96}. While simple and space-efficient, its linear depth becomes a limiting factor for large inputs, especially in the context of fault-tolerant quantum computing. To address the depth bottleneck, quantum carry-lookahead adders were proposed, starting with Draper et al. \cite{DKR06}. These adders reduce the circuit depth down to logarithmic, at the cost of using more workspace and losing the nearest-neighbor connectivity that ripple-carry adders have. These adders are particularly attractive for near-term architectures where minimizing coherence time is critical.

Until now, quantum ripple-carry and quantum carry-lookahead adders have generally been treated as fundamentally distinct approaches in quantum circuit design, primarily due to a lack of understanding regarding their underlying connections. Ripple-carry addition has been viewed as a simpler, sequential architecture, while carry-lookahead addition is seen as a more complex but highly parallelizable alternative. This perceived separation stems from limited insights into how the two models might be derived from a common framework. As a result, research has typically focused on comparing their characteristics rather than exploring potential structural or conceptual unification. 

However, in a work presented earlier this year by Remaud and Vandaele \cite{RV25}, a new method for addition was introduced, based on a new way to implement ladders of Toffoli gates. This technique does not require ancillary qubits, just like the ripple-carry technique, and has sublinear depth, just like the carry-lookahead technique. All this comes at the cost of using an increased number of gates. A Venn diagram is provided in Figure~\ref{fig:venn} to visualize the different properties of these different techniques.

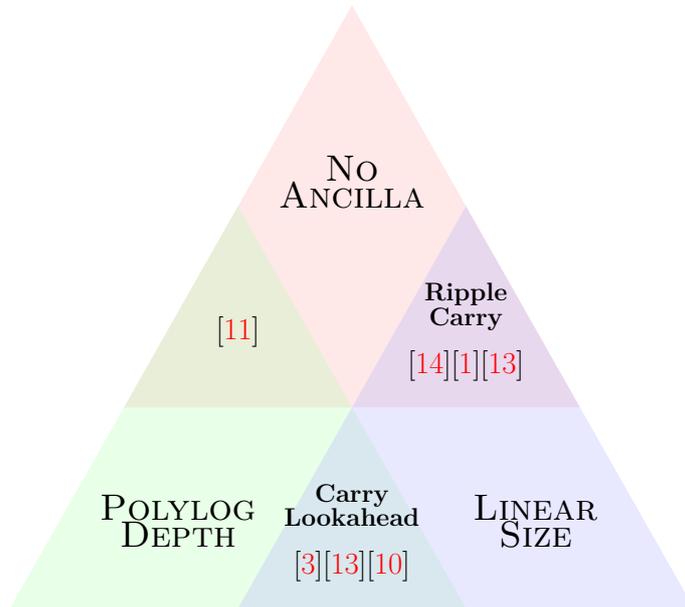
\begin{figure}[H]
    \centering
    \resizebox{0.6\columnwidth}{!}{%
    \begin{tikzpicture}[thick, every text node part/.style={align=center}]
      \fill[red!30, opacity=0.3] (-4.5,0) -- (4.5,0) -- (0,8) -- cycle;
      \node at (0,4.5) {\huge \textsc{No} \\ \huge \textsc{Ancilla}};
    
      \fill[green!30, opacity=0.3] (-2.25,4) -- (-6.75,-4) -- (2.25,-4) -- cycle;
      \node at (-1.75,-2.25) [anchor=east] {\huge \textsc{Polylog} \\ \huge \textsc{Depth}};
    
      \fill[blue!30, opacity=0.3] (2.25,4) -- (-2.25,-4) -- (6.75,-4) -- cycle;
      \node at (2.25,-2.25) [anchor=west] {\huge \textsc{Linear} \\ \huge \textsc{Size}};
    
      \node at (-2.25,1.5) {\LARGE \cite{RV25}};
      \node at (2.25,1.5) {\Large \textbf{Ripple} \\ \Large \textbf{Carry} \\ \\ \LARGE \cite{VBE96}\cite{CDKM04}\cite{TTK10}};
      \node at (0,-2.5) {\Large \textbf{Carry} \\ \Large \textbf{Lookahead} \\ \\ \LARGE \cite{DKR06}\cite{TTK10}\cite{Mog19}};
    
    \end{tikzpicture}}
    \caption{Venn diagram of in-place quantum reversible adders with classical logic only.}\label{fig:venn}
\end{figure}

Each of these types of adder offers unique advantages and limitations, and their applicability often depends on the broader algorithmic and architectural context. Table \ref{table:AllAdders} provides an overview of the complexities of different algorithms that have been proposed to implement in-place addition.

\begin{table}[H]
    \centering
    \caption{Asymptotic complexity of in-place quantum reversible adders with classical logic only.}\label{table:AllAdders}
    \begin{tabular}{|c|c|c|c|c|} \hline
        Paper           & Toffoli count         & Toffoli Depth         & Ancilla   \\ \hline
        \cite{VBE96}    & $4n-2$                & $3n-1$                & $n$ \\
        \cite{CDKM04}   & $2n-1$                & $2n-1$                & 1 \\
        \cite{TTK10}    & $2n-1$                & $2n-1$                & 0 \\ \hline
        \cite{RV25}     & $\bigO{n \log n}$     & $\bigO{\log^2 n}$     & 0 \\ \hline
        \cite{TTK10}    & $14n + \Theta(1)$      & $18 \log n + \Theta(1)$           & $3 n / \log n  + \Theta(1)$ \\
        \cite{Mog19}    & $12n + \Theta(\log n)$ & $10 \log n + \Theta(1)$& $n-1$ \\
        \cite{DKR06}    & $10n - \Theta(\log n)$ & $4 \log n + \Theta(1)$ & $2n - \Theta(\log n)$ \\
        This paper      & $8n - \Theta(\log n)$ & $4 \log n + \Theta(1)$ & $n - \Theta(\log n)$ \\ \hline
    \end{tabular}
\end{table}

\textbf{Our contributions.} We give in Section~\ref{sec:preli} preliminaries and notation, before technical details in the subsequent sections.
\begin{itemize}
    \item In Section~\ref{sec:ladders}, we take a closer look at how ladders of Toffoli gates can be implemented. Currently, there are three distinct implementations: the first is naive, with linear depth, the second, proposed by \cite{RV25}, has polylogarithmic depth, and the third has logarithmic depth and was implicitly used by \cite{DKR06}. 
    \item In Section~\ref{sec:CLA}, we show that there are two main structures for performing addition which are shared by several existing adders, and within which the main subroutine is the Toffoli ladder. The "original structure" (implicitly used by \cite{VBE96} and \cite{DKR06}) uses a linear number of ancilla qubits, while the second, the "space-optimized structure" (implicitly used by \cite{TTK10} and \cite{RV25}), does not use any. Based on this observation, we note that it is possible to design a new adder by embedding the logarithmic depth implementation of the Toffoli ladder in the second structure.
\end{itemize}

%% file: 1-preliminaries.tex
\section{Preliminaries} \label{sec:preli}

We recall the definitions of the operators discussed in this paper: ladders and adders.

\subsection{Ladders}

We begin with the definition of the $\CNOT$ ladder \cite{RV25}.

\begin{definition}
    Let $x_i \in \{0,1\} \; \forall i \in [\![0,n]\!]$ and $X$ denote the quantum register $\bigotimes_{i=0}^{n} \ket{x_i}$. We define $\ladder_1$ on $n+1$ qubits as the operator $\Lg_1^{(n)}$ with the following action:
    \begin{equation*}\label{eq:ladder1}
        \Lg_1^{(n)} \Rbra{X}  \eqdef \ket{x_0} \otimes \Rbra{\bigotimes_{i=1}^{n} \ket{x_i \oplus x_{i-1}}}
    \end{equation*}
\end{definition}

We also recall the definition of the Toffoli ladder \cite{RV25}.

\begin{definition}
    Let $x_i \in \{0,1\} \; \forall i \in [\![0,n]\!]$ and $y_i \in \{0,1\} \; \forall i \in [\![0,n-1]\!]$. Let $X$ and $Y$, respectively, denote the quantum registers $\bigotimes_{i=0}^{n} \ket{x_i}$ and $\bigotimes_{i=0}^{n-1} \ket{y_i}$. We define $\ladder_2$ on $2n+1$ qubits as the operator $\Lg_2^{(n)}$ with the following action:
    \begin{equation*}\label{eq:ladder2}
        \Lg_2^{(n)} \Rbra{X,Y} \eqdef \ket{x_0} \otimes \Rbra{\bigotimes_{i=0}^{n-1} \ket{x_{i+1} \oplus x_i y_i}} \otimes Y
    \end{equation*}
\end{definition}

\subsection{Addition}

We will work with two $n$-bit numbers denoted $a$ and $b$. We compute the addition in-place, meaning that we want an operator $\Add_n$ with the following action:
\begin{equation*}
	\ket{a} \ket{b} \ket{z} \xmapsto{\Add_n} \ket{a} \ket{a + b \mod 2^n} \ket{z \oplus (a + b)_n}
\end{equation*}
(where $z \in \{0,1\}$) using only gates from the set $\{\text{Toffoli}, \CNOT, \Xg\}$.
     
Note that we consider only reversible implementations, and if we have to use ancilla qubits, they have to be reset to zero at the end of the circuit.

\subsection{Notation}

Throughout this document, $\log x$ will denote the binary logarithm of $x$.
Inside circuits, slice numbers refer to the block preceding them.

%% file: 2-ladders.tex
\section{Implementation of the Toffoli ladder}\label{sec:ladders}

In this section, we give an overview of the implementations for the $\ladder_2$ operator and their complexity. 

\subsection{Linear depth}

The first implementation is the most straightforward one and gives this operator its name. It literally takes the form of a ladder of $n$ Toffoli gates to implement $\Lg_2^{(n)}$. Figure~\ref{circ:LinLad} shows the circuit resulting from this naive implementation for $n=7$.

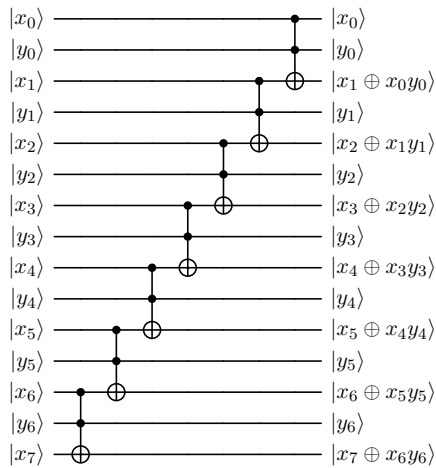
\begin{figure}[H]
    \centering
    \scalebox{0.75}{\input{images/LadderLin}}
    \caption{Linear depth implementation of the operator $\Lg_2^{(7)}$.}\label{circ:LinLad}
\end{figure}

In a very straightforward manner, we can establish Lemma~\ref{lem:LadLin}, which gives us the complexity associated with this implementation.

\begin{lemma}\label{lem:LadLin}
    There exists a Toffoli circuit that implements $\Lg_2^{(n)}$ with a Toffoli-depth of $n$ and a Toffoli-count of $n$, without any ancilla qubit.
\end{lemma}

\subsection{Polylogarithmic depth}

In a recent paper \cite{RV25}, it has been proven that it is possible to construct a circuit that is asymptotically much shallower, also without using ancilla qubits, at the cost of increasing the number of gates and having increased connectivity.
We give an example of decomposition for $n=7$ in Figure~\ref{circ:PolylogLad}. It should be noted that this decomposition uses the decomposition of multi-controlled $\Xg$ gates in logarithmic time \cite{KG24} as a subroutine. Thus, the figure given here as an example does not represent what would actually be implemented (the decomposition of 3- and 5-control gates would involve using a considerable number of gates) but represents what happens on a larger scale: the decomposition of the operator $\Lg_2^{(n)}$ into a circuit with $\bigO{\log n}$ time slices containing gates that can be implemented in Toffoli-depth of the order of $\bigO{\log n}$.

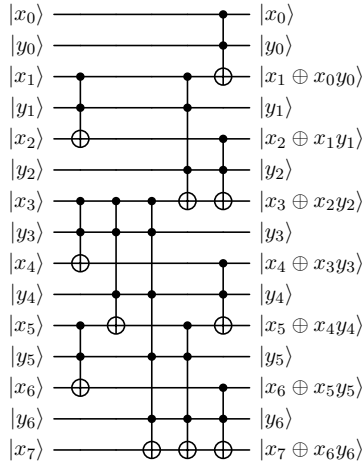
\begin{figure}[H]
    \centering
    \scalebox{0.75}{\input{images/LadderPolylog}}
    \caption{Polylogarithmic depth implementation of the operator $\Lg_2^{(7)}$ \cite{RV25}.} \label{circ:PolylogLad}
\end{figure}

We reproduce in Lemma~\ref{lem:LadPolylog} the result demonstrated in \cite{RV25}, and refer to that paper for further details. 

\begin{lemma}[Lemma 4 in \cite{RV25}]\label{lem:LadPolylog}
    There exists a circuit that implements $\Lg_2^{(n)}$ over the $\{\text{Toffoli}, \Xg\}$ gate set with a depth of $\bigO{\log^2 n}$ and a gate count of $\bigO{n \log n}$, without any ancilla qubit.
\end{lemma}

\subsection{Logarithmic depth}

Finally, it should be noted that there is a method with logarithmic depth for implementing $\ladder_2$, using ancilla qubits. It was used two decades ago in an article by Draper $\etal$ \cite{DKR06}, but to the best of our knowledge, it was not explicitly identified as such in that article or in subsequent works. The (dagger version of this) method is called CARRY in \cite{TTK10} and is not named in the original work by Draper $\etal$, but corresponds to the (dagger version of the) algorithm described in their Section 3 and consists of P-, G-, C- and P$^{-1}$- rounds. We give an example of this circuit in Figure~\ref{circ:LogLad} for $n=7$.

\begin{figure}[H]
    \centering
    \scalebox{0.75}{\input{images/LadderLog}}
    \caption{Logarithmic depth implementation of the operator $\Lg_2^{(7)}$ \cite{DKR06}.} \label{circ:LogLad}
\end{figure}
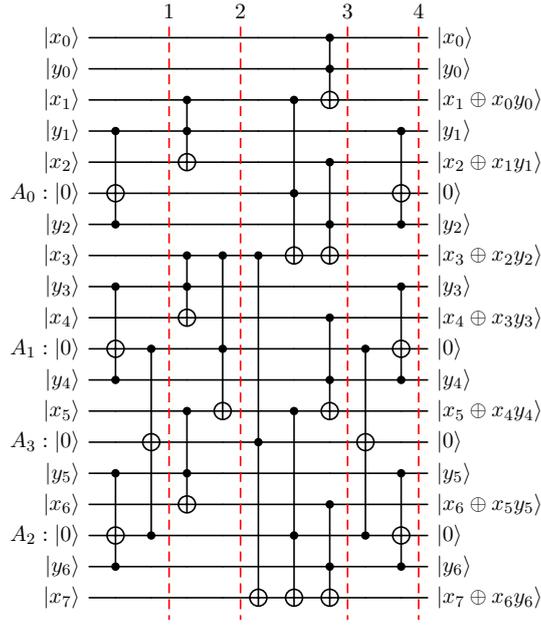

We give in Algorithm~\ref{alg:pgcp} the corresponding pseudocode for any $n$, where we defined $\sigma(i) = n - i - 2\floor{\frac{n}{2^i}} - \omega(n \mod 2^i)$ to facilitate the writing of the indexes.

\begin{algorithm}[H]
	\caption{$\CARRY^\dag$ a.k.a. $\Lg_2^{(n-1)}$} \label{alg:pgcp}
	\begin{algorithmic}[1]
		\Require $\ket{a}_A \ket{b}_B$ where $a \in \{0,1\}^{n}$ and $b \in \{0,1\}^{n-1}$   
		\Ensure $\Lg_2^{(n-1)}\Rbra{A,B}$ using a register $C$ of $n - \hw{n} - \floor{\log n}$ ancilla qubits

        \For{$j = 1$ to $\floor{\frac{n}{2}} - 1$} \Comment{Slice 1}
            \State $\CCNOT( B_{2j-1}, B_{2j}, C_{j-1} )$
        \EndFor
		\For{$i = 2$ to $\floor{\log{n}} - 1$}
    		\For{$j = 1$ to $\floor{\frac{n}{2^i}} - 1$}
        		\State $\CCNOT( C_{2j + \sigma(i-1)}, C_{2j + 1 + \sigma(i-1)}, C_{j + \sigma(i)} )$
    		\EndFor
		\EndFor
        
        \For{$j = 1$ to $\floor{\frac{n - 1}{2}}$} \Comment{Slice 2}
            \State $\CCNOT( A_{2j-1}, B_{2j-1}, A_{2j} )$
        \EndFor
		\For{$i = 2$ to $\floor{\log{\frac{2n}{3}}}$}
    		\For{$j = 1$ to $\floor{\frac{n - 2^{i-1}}{2^i}}$}
        		\State $\CCNOT( A_{2^i j - 1}, C_{2j + \sigma(i-1)}, A_{2^i j + 2^{i - 1} - 1} )$
    		\EndFor
		\EndFor
        
		\For{$i = \floor{\log{n}}$ to $2$} \Comment{Slice 3}
    		\For{$j = 1$ to $\floor{\frac{n}{2^i}}$}
                \State $\CCNOT( A_{2^i j - 2^{i-1} - 1}, C_{2 j - 1 + \sigma(i-1)}, A_{2^i j  - 1} )$
    		\EndFor
		\EndFor
        \For{$j = 1$ to $\floor{\frac{n}{2}}$}
            \State $\CCNOT( A_{2j-2}, B_{2j-2}, A_{2j-1} )$
        \EndFor
		
        \State Uncompute Slice 1 \Comment{Slice 4}
	\end{algorithmic}
\end{algorithm}

The complexity of Algorithm~\ref{alg:pgcp} is given in Lemma~\ref{lem:LadLog}.

\begin{lemma}[Section 3 in \cite{DKR06}]\label{lem:LadLog}
    There exists a Toffoli circuit that implements $\Lg_2^{(n-1)}$ with a depth of $\floor{\log n} + \floor{\log \frac{n}{3}} + 3$ and a Toffoli count of $4n - 3\hw{n} - 3\floor{\log n} - 1$, with $n - \hw{n} - \floor{\log n}$ ancilla qubits.
\end{lemma}
\begin{proof}
    The closed formulas come from the paper by Draper $\etal$ \cite{DKR06}.

    We have experimentally verified with Q-Pragma \cite{GK24} that this algorithm effectively implements the $\ladder_2$ operator, which can also be easily verified using the following simple substitution pattern:
    \begin{equation*}
    \scalebox{0.7}{
        \begin{quantikz}[column sep=.3cm, row sep={.7cm,between origins}]
            &           & \ctrl{2}  &       & \\
            & \ctrl{1}  &           & \ctrl{1}  & \\
            & \targ{}   & \ctrl{2}  & \targ{}   & \\
            & \ctrl{-1} &           & \ctrl{-1} & \\
            &           & \targ{}   &           &
        \end{quantikz}
        =
        \begin{quantikz}[column sep=.3cm, row sep={.7cm,between origins}]
            & \ctrl{1}  & \\
            & \ctrl{2}  & \\
            &           & \\
            & \ctrl{1}  & \\
            & \targ{}   &
        \end{quantikz}}
    \end{equation*}
    to directly incorporate the first and last rounds (P and P$^{-1}$ rounds) into the two middle rounds (C and G rounds) via multi-controlled $\Xg$ gates. The ancilla qubits can thus be discarded, bringing us back to the polylogarithmic depth construction described above.
\end{proof}

%% file: images/LadderLin.tex
\begin{quantikz}[column sep=.3cm, row sep={.55cm,between origins}]
    \lstick{$\ket{x_{0}}$}	&			&			&			&			&			&			& \ctrl{1}	& \rstick{$\ket{x_{0}}$}	\\
    \lstick{$\ket{y_{0}}$}	&			&			&			&			&			&			& \ctrl{1}	& \rstick{$\ket{y_{0}}$}	\\
    \lstick{$\ket{x_{1}}$}	&			&			&			&			&			& \ctrl{1}	& \targ{}	& \rstick{$\ket{x_{1} \oplus x_{0} y_{0}}$}	\\
    \lstick{$\ket{y_{1}}$}	&			&			&			&			&			& \ctrl{1}	&			& \rstick{$\ket{y_{1}}$}		\\
    \lstick{$\ket{x_{2}}$}	&			&			&			&			& \ctrl{1}	& \targ{}	&			& \rstick{$\ket{x_{2} \oplus x_{1} y_{1}}$}	\\
    \lstick{$\ket{y_{2}}$}	&			&			&			&			& \ctrl{1}	&			&			& \rstick{$\ket{y_{2}}$}		\\
    \lstick{$\ket{x_{3}}$}	&			&			&			& \ctrl{1}	& \targ{}	&			&			& \rstick{$\ket{x_{3} \oplus x_{2} y_{2}}$}	\\
    \lstick{$\ket{y_{3}}$}	&			&			&			& \ctrl{1}	&			&			&			& \rstick{$\ket{y_{3}}$}		\\
    \lstick{$\ket{x_{4}}$}	&			&			& \ctrl{1}	& \targ{}	&			&			&			& \rstick{$\ket{x_{4} \oplus x_{3} y_{3}}$}	\\
    \lstick{$\ket{y_{4}}$}	&			&			& \ctrl{1}	&			&			&			&			& \rstick{$\ket{y_{4}}$}		\\
    \lstick{$\ket{x_{5}}$}	&			& \ctrl{1}	& \targ{}	&			&			&			&			& \rstick{$\ket{x_{5} \oplus x_{4} y_{4}}$}	\\
    \lstick{$\ket{y_{5}}$}	&			& \ctrl{1}	&			&			&			&			&			& \rstick{$\ket{y_{5}}$}		\\
    \lstick{$\ket{x_{6}}$}	& \ctrl{1}	& \targ{}	&			&			&			&			&			& \rstick{$\ket{x_{6} \oplus x_{5} y_{5}}$}	\\
    \lstick{$\ket{y_{6}}$}	& \ctrl{1}	&			&			&			&			&			&			& \rstick{$\ket{y_{6}}$}		\\
    \lstick{$\ket{x_{7}}$}	& \targ{}	&			&			&			&			&			&			& \rstick{$\ket{x_{7} \oplus x_{6} y_{6}}$}
\end{quantikz}

%% file: images/LadderPolylog.tex
\begin{quantikz}[column sep=.3cm, row sep={.55cm,between origins}]
    \lstick{$\ket{x_{0}}$}	&			&			&			&           & \ctrl{1}	& \rstick{$\ket{x_{0}}$}					\\
    \lstick{$\ket{y_{0}}$}	&			&			&			&           & \ctrl{1}	& \rstick{$\ket{y_{0}}$}					\\
    \lstick{$\ket{x_{1}}$}	& \ctrl{1}  &			&			& \ctrl{1}	& \targ{}	& \rstick{$\ket{x_{1} \oplus x_{0} y_{0}}$} \\
    \lstick{$\ket{y_{1}}$}	& \ctrl{1}  &			&			& \ctrl{2}	&			& \rstick{$\ket{y_{1}}$}					\\
    \lstick{$\ket{x_{2}}$}	& \targ{}	&           &           &       	& \ctrl{1}  & \rstick{$\ket{x_{2} \oplus x_{1} y_{1}}$}	\\
    \lstick{$\ket{y_{2}}$}	&			&           &			& \ctrl{1}  & \ctrl{1}  & \rstick{$\ket{y_{2}}$}	\\
    \lstick{$\ket{x_{3}}$}	& \ctrl{1}	& \ctrl{1}  & \ctrl{1}  & \targ{}   & \targ{}	& \rstick{$\ket{x_{3} \oplus x_{2} y_{2}}$}	\\
    \lstick{$\ket{y_{3}}$}	& \ctrl{1}	& \ctrl{2}	& \ctrl{2}  &           &			& \rstick{$\ket{y_{3}}$}	\\
    \lstick{$\ket{x_{4}}$}	& \targ{}	&       		&           &           & \ctrl{1}	& \rstick{$\ket{x_{4} \oplus x_{3} y_{3}}$}	\\
    \lstick{$\ket{y_{4}}$}	&			& \ctrl{1}  & \ctrl{2}	&           & \ctrl{1}	& \rstick{$\ket{y_{4}}$}	\\
    \lstick{$\ket{x_{5}}$}	& \ctrl{1}	& \targ{}   &			& \ctrl{1}  & \targ{}	& \rstick{$\ket{x_{5} \oplus x_{4} y_{4}}$}	\\
    \lstick{$\ket{y_{5}}$}	& \ctrl{1}	&			& \ctrl{2}	& \ctrl{2}  &			& \rstick{$\ket{y_{5}}$}	\\
    \lstick{$\ket{x_{6}}$}	& \targ{}	&           &           &           & \ctrl{1}	& \rstick{$\ket{x_{6} \oplus x_{5} y_{5}}$}	\\
    \lstick{$\ket{y_{6}}$}	&           &       	& \ctrl{1}	& \ctrl{1}  & \ctrl{1}	& \rstick{$\ket{y_{6}}$}	\\
    \lstick{$\ket{x_{7}}$}	&           &			& \targ{}   & \targ{}   & \targ{}	& \rstick{$\ket{x_{7} \oplus x_{6} y_{6}}$}
\end{quantikz}

%% file: images/LadderLog.tex
\begin{quantikz}[column sep=.3cm, row sep={.55cm,between origins}]
    \lstick{$\ket{x_{0}}$}	&			& \slice{1}	&       	& \slice{2} &			&           & \ctrl{1}\slice{3}	&   &  \slice{4} & \rstick{$\ket{x_{0}}$}					\\
    \lstick{$\ket{y_{0}}$}	&			&       	&       	&			&			&           & \ctrl{1}	&       	&       	& \rstick{$\ket{y_{0}}$}					\\
    \lstick{$\ket{x_{1}}$}	&			&       	& \ctrl{1}	&			&			& \ctrl{3}	& \targ{}	&       	&       	& \rstick{$\ket{x_{1} \oplus x_{0} y_{0}}$} \\
    \lstick{$\ket{y_{1}}$}	& \ctrl{2}  &       	& \ctrl{1}	&			&			&			&			&       	& \ctrl{2}  & \rstick{$\ket{y_{1}}$}					\\
    \lstick{$\ket{x_{2}}$}	&			&       	& \targ{}	&           &           &       	& \ctrl{2}  &       	&       	& \rstick{$\ket{x_{2} \oplus x_{1} y_{1}}$}	\\
    \lstick{$A_0: \ket{0}$}	& \targ{}	&       	&       	&			&			& \ctrl{2}  &			&       	& \targ{}	& \rstick{$\ket{0}$} \\
    \lstick{$\ket{y_{2}}$}	& \ctrl{-1}	&       	&       	&           &			&			& \ctrl{1}  &       	& \ctrl{-1}	& \rstick{$\ket{y_{2}}$}	\\
    \lstick{$\ket{x_{3}}$}	&			&       	& \ctrl{1}	& \ctrl{3}  & \ctrl{6}  & \targ{}   & \targ{}	&       	&       	& \rstick{$\ket{x_{3} \oplus x_{2} y_{2}}$}	\\
    \lstick{$\ket{y_{3}}$}	& \ctrl{2}	&       	& \ctrl{1}	&			&			&           &			&       	& \ctrl{2}  & \rstick{$\ket{y_{3}}$}	\\
    \lstick{$\ket{x_{4}}$}	&			&       	& \targ{}	&       	&           &           & \ctrl{2}	&       	&       	& \rstick{$\ket{x_{4} \oplus x_{3} y_{3}}$}	\\
    \lstick{$A_1: \ket{0}$}	& \targ{}	& \ctrl{3}	&       	& \ctrl{2}	&			&			&			& \ctrl{3}	& \targ{}	& \rstick{$\ket{0}$} \\
    \lstick{$\ket{y_{4}}$}	& \ctrl{-1}	&       	&       	&			&			&           & \ctrl{1}	&       	& \ctrl{-1}	& \rstick{$\ket{y_{4}}$}	\\
    \lstick{$\ket{x_{5}}$}	&			&       	& \ctrl{2}	& \targ{}   &			& \ctrl{4}  & \targ{}	&       	&       	& \rstick{$\ket{x_{5} \oplus x_{4} y_{4}}$}	\\
	\lstick{$A_3: \ket{0}$}	&			& \targ{}	& 			&    		& \ctrl{5}	&			&			& \targ{}	&       	& \rstick{$\ket{0}$} \\
    \lstick{$\ket{y_{5}}$}	& \ctrl{2}	&       	& \ctrl{1}	&			&			&			&			&       	& \ctrl{2}  & \rstick{$\ket{y_{5}}$}	\\
    \lstick{$\ket{x_{6}}$}	&			&       	& \targ{}	&           &           &           & \ctrl{2}	&       	&       	& \rstick{$\ket{x_{6} \oplus x_{5} y_{5}}$}	\\
	\lstick{$A_2: \ket{0}$}	& \targ{}	& \ctrl{-3}	&       	&  			&			& \ctrl{2}	&			& \ctrl{-3}	& \targ{}	& \rstick{$\ket{0}$} \\
    \lstick{$\ket{y_{6}}$}	& \ctrl{-1}	&       	&       	&       	&			&			& \ctrl{1}	&       	& \ctrl{-1}	& \rstick{$\ket{y_{6}}$}	\\
    \lstick{$\ket{x_{7}}$}	&           &       	&       	& 			& \targ{}   & \targ{}   & \targ{}	&       	&       	& \rstick{$\ket{x_{7} \oplus x_{6} y_{6}}$}
\end{quantikz}

%% file: 3-CLA.tex
\section{New Quantum Carry-Lookahead Adder}\label{sec:CLA}

We show here that some of the earliest quantum adders historically proposed are linked by an underlying structure, with the only subroutine differentiating them being the one used to implement the $\ladder_2$ operator.

More specifically, we can see that two different general structures have been adopted for building adders. One uses $n-1$ ancilla qubits and was used to build (an equivalent version of) the first ripple-carry adder \cite{VBE96} as well as the first carry-lookahead adder \cite{DKR06}. The second structure does not require these ancilla qubits and is the basis for (an equivalent version of) the arguably most optimized ripple-carry adder in terms of several metrics \cite{TTK10} as well as the first adder to have sublinear depth and no ancilla qubits \cite{RV25}, proposed earlier this year.

Within these two structures, we find the $\ladder_2$ subroutine, which, as presented in the previous section, can be implemented in three different ways: with linear depth (Lemma~\ref{lem:LadLin}), polylogarithmic depth (Lemma~\ref{lem:LadPolylog}) and logarithmic depth (Lemma~\ref{lem:LadLog}). For example, by taking the structure that requires a lot of ancilla qubits and implementing Toffoli ladders with the logarithmic depth construction, we obtain Draper $\etal$'s carry-lookahead adder \cite{DKR06}.

The various resulting combinations of ‘adder structure / ladder implementation’ are given in Table~\ref{table:Adders}. Four of the six have already been proposed in the literature (or equivalent versions). Two are new, with one being of particular interest. We examine it in this section.

\begin{table}[H]
    \centering
    \caption{Source of the different adders obtained by combining one of the two adder structures with an implementation for the Toffoli ladders.}\label{table:Adders}
    \begin{tabular}{|c||c|c|c|} \hline
                                    & Lemma~\ref{lem:LadLin}    & Lemma~\ref{lem:LadPolylog}    & Lemma~\ref{lem:LadLog}        \\ \hline \hline
        Algorithm~\ref{alg:main1}   & $\approx$ \cite{VBE96}    & Remark~\ref{prop:1}      & \cite{DKR06}                  \\ \hline
        Algorithm~\ref{alg:main2}   & $\approx$ \cite{TTK10}    & \cite{RV25}                   & Theorem~\ref{prop:2}      \\ \hline
    \end{tabular}
\end{table}

\subsection{The "original" structure}

Whether it is the first ripple-carry adder \cite{VBE96} or the first carry-lookahead adder \cite{DKR06}, both have the same structure, using a linear number of ancilla qubits by default. This structure is provided by Algorithm~\ref{alg:main1}. The only subroutine that can be implemented in different ways is $\ladder_2$.

\begin{algorithm}[H]
	\caption{Structure for implementing $\Add_n$ using ancilla qubits}\label{alg:main1}
	\begin{algorithmic}[1]
		\Require $\ket{a}_{A} \ket{b}_{B} \ket{0^{\otimes (n-1)}, z}_{C}$ where $a,b \in \ints{0,2^n - 1}$ and $z \in \{0,1\}$ is stored in the last qubit in the ancillary register $C$.
		\Ensure $\ket{a}_{A} \ket{a + b \mod 2^n}_{B} \ket{0^{\otimes (n-1)}, z \oplus (a + b)_n}_{C}$.
		
		\For{$i = 0$ to $n - 1$} \Comment{Slice 1}
    		\State $\CCNOT(A_i, B_i, C_i)$ 
    		\State $\CNOT(A_i, B_i)$ 
    	\EndFor
    	
    	\State $\Rbra{\Lg_2^{(n-1)}}^\dag(C,B[1 \colon])$ \Comment{Slice 2}
		
		\For{$i = 0$ to $n - 2$} \Comment{Slice 3}
			\State $\CNOT(C_{i}, B_{i + 1})$
                \State $\CNOT(A_{i}, B_{i})$
                \State $\NOT(B_{i})$
		\EndFor
		
		\State $\Lg_2^{(n-2)}(C[\colon -1],B[1 \colon\colon -1])$ \Comment{Slice 4}
		
		\For{$i = 0$ to $n - 2$} \Comment{Slice 5}
			\State $\CNOT(A_i, B_i)$
			\State $\CCNOT(A_i, B_i, C_i)$
                \State $\NOT(B_i)$
		\EndFor
	\end{algorithmic}
\end{algorithm}

When the logarithmic depth implementation for $\ladder_2$ is used, we fall directly back on the adder of Draper $\etal$ \cite{DKR06}. To find the adder of Vedral $\etal$ \cite{VBE96} when the linear depth implementation is used instead, the equality given in Figure~\ref{fig:VBE} is needed.

\begin{figure}[H]
    \centering
    \scalebox{0.7}{
    \begin{quantikz}[column sep=.3cm, row sep={.7cm,between origins}]
        \lstick{$C_{i-1}$} & \ctrl{2}  &           &           &           & \ctrl{2}  & \\
        \lstick{$A_i$} &           & \ctrl{1}  & \ctrl{1}  & \ctrl{1}  &           & \\
        \lstick{$B_i$} & \ctrl{1}  & \targ{}   & \ctrl{1}  & \targ{}   & \targ{}   & \\
        \lstick{$C_i$} & \targ{}   &           & \targ{}   &           &           &
    \end{quantikz}
    =
    \begin{quantikz}[column sep=.3cm, row sep={.7cm,between origins}]
        & \ctrl{2}  &           &           & \ctrl{2}  &           &           &           & \\
        &           & \ctrl{1}  &           &           & \ctrl{1}  & \ctrl{1}  &           & \\
        & \targ{}   & \targ{}   & \targ{}   & \ctrl{1}  & \targ{}   & \ctrl{1}  & \targ{}   & \\
        &           &           &           & \targ{}   &           & \targ{}   &           &
    \end{quantikz}}
    \caption{On the left, the subroutine used in Vedral $\etal$'s paper. On the right, the equivalent circuit used in the structure described in Algorithm~\ref{alg:main1}.}\label{fig:VBE}
\end{figure}
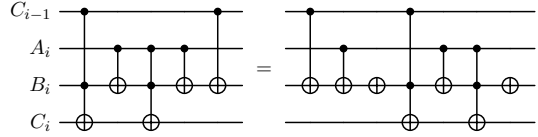

Finally, we can use the polylogarithmic depth implementation of $\ladder_2$ (Lemma~\ref{lem:LadPolylog}) and obtain a new adder. However, it does not have any particularly interesting properties considering the state-of-the-art, and we mention it for the sake of completeness in Remark~\ref{prop:1}.

\begin{remark}\label{prop:1}
    There exists a circuit implementing the operator $\Add_n$ with $n-1$ ancilla qubits, that has a Toffoli count of $\bigO{n \log n}$ and a Toffoli depth of $\bigO{\log^2 n}$.
\end{remark}

\subsection{The space optimized structure}

 Takahashi $\etal$ \cite{TTK10} implicitly used another structure to propose a ripple-carry adder (and therefore used $\ladder_2$'s naive implementation) which does not use ancilla qubits. It was also recently adopted by Remaud and Vandaele \cite{RV25} with the polylogarithmic implementation for $\ladder_2$. This structure is provided by Algorithm ~\ref{alg:main2}.

\begin{algorithm}[H]
    \caption{Structure for implementing $\Add_n$ using no ancilla qubit}\label{alg:main2}
    \begin{algorithmic}[1]
        \Require $\ket{a}_{A} \ket{b}_{B} \ket{z}_{Z}$ where $a,b \in \ints{0,2^n - 1}$ are respectively stored in the registers $A$ and $B$, and $z \in \{0,1\}$ is stored in a qubit $Z$.
        \Ensure $\ket{a}_{A} \ket{a + b \mod 2^n}_{B} \ket{z \oplus (a + b)_n}_{Z}$.
        
        \For{$i = 1$ to $n - 1$} \Comment{Slice 1}
            \State $\CNOT(A_i, B_i)$ 
        \EndFor
    
        \State $\Lg_1^{(n-1)}(A[1 \colon],Z)$ \Comment{Slice 2}
        
        \State $\Rbra{\Lg_2^{(n)}}^\dag(A,Z,B)$ \Comment{Slice 3}
        
        \For{$i = 1$ to $n - 1$} \Comment{Slice 4}
            \State $\CNOT(A_i, B_i)$ 
        \EndFor
    
        \For{$i = 1$ to $n - 2$}
            \State $\Xg(B_i)$ 
        \EndFor
        
        \State $\Lg_2^{(n-1)}(A,B[\colon -1])$ \Comment{Slice 5}

        \State $\Rbra{\Lg_1^{(n-2)}}^\dag(A[1 \colon])$ \Comment{Slice 6}
        
        \For{$i = 0$ to $n - 1$} \Comment{Slice 7}
            \State $\CNOT(A_i, B_i)$
        \EndFor
        \For{$i = 1$ to $n - 2$}
            \State $\Xg(B_i)$ 
        \EndFor
    \end{algorithmic}
\end{algorithm}

Here, we use $\ladder_1$ and $\ladder_2$. The first can be implemented naively with linear $\CNOT$-depth, or it can be implemented with a linear number of $\CNOT$ gates in logarithmic $\CNOT$ depth, as stated in the following Lemma.

\begin{lemma}[Lemma 2 in \cite{RV25}]\label{lem:CNOTladder}
    Let $n \ge 2$ be an integer. The operator $\Lg_1^{(n)}$ can be implemented with a $\CNOT$-depth of $2 \log n + \Theta(1)$ and a $\CNOT$-count of $2n + \Theta(1)$.
\end{lemma}

When the polylogarithmic depth implementation for $\ladder_2$ is used, we fall directly back on the adder of Remaud and Vandaele \cite{RV25}. To find the adder of Takahashi $\etal$ \cite{TTK10} when the linear depth implementation is used instead, the equality given in Figure~\ref{fig:TTK} is needed.

\begin{figure}[H]
    \centering
    \scalebox{0.7}{
    \begin{quantikz}[column sep=.3cm, row sep={.7cm,between origins}]
        \lstick{$A_i$}      & \ctrl{1}  & \ctrl{1}  & \\
        \lstick{$B_i$}      & \ctrl{1}  & \targ{}   & \\
        \lstick{$A_{i+1}$}  & \targ{}   &           &
    \end{quantikz}
    =
    \begin{quantikz}[column sep=.3cm, row sep={.7cm,between origins}]
        & \ctrl{1}  &           & \ctrl{1}  &          & \\
        & \targ{}   & \targ{}   & \ctrl{1}  & \targ{}  & \\
        &           &           & \targ{}   &          &
    \end{quantikz}}
    \caption{On the left, the subroutine used in Takahashi $\etal$'s paper. On the right, the equivalent circuit used in the structure described in Algorithm~\ref{alg:main2}.}\label{fig:TTK}
\end{figure}
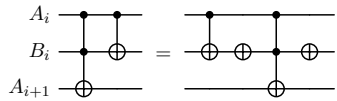

Finally, we can use the logarithmic depth implementation of $\ladder_2$ (Lemma~\ref{lem:LadLog}) and obtain a new adder. We state its properties in Theorem~\ref{prop:2}.

\begin{theorem}\label{prop:2}
    There exists a circuit implementing the operator $\Add_n$ with $n-\hw{n}-\floor{\log n}$ ancilla qubits, which has a Toffoli-count of $8n - \Theta(\log n)$ and a Toffoli-depth of $4 \log{n} + \Theta(1)$.
\end{theorem}
\begin{proof}
    Slices 1, 4 and 7 are implemented in constant $\CNOT$-depth with a total of $3n+\Theta(1)$ $\CNOT$ gates. Slices 2 and 6 are implemented using Lemma~\ref{lem:CNOTladder}, $\ie$, with a total of $4n+\Theta(1)$ $\CNOT$ gates and a $\CNOT$-depth of $4 \log n + \Theta(1)$.
    Finally, the $\ladder_2$ operators in Slices 3 and 6 are implemented using Lemma~\ref{lem:LadLog}, $\ie$, with a total of $8n - \bigO{\log n}$ Toffoli gates and a Toffoli-depth of $4 \log n + \Theta(1)$, at the expense of using $n-\hw{n}-\floor{\log n}$ ancilla qubits.
\end{proof}

An example of circuit for $n=8$ is provided in Figure~\ref{fig:newadder}. 

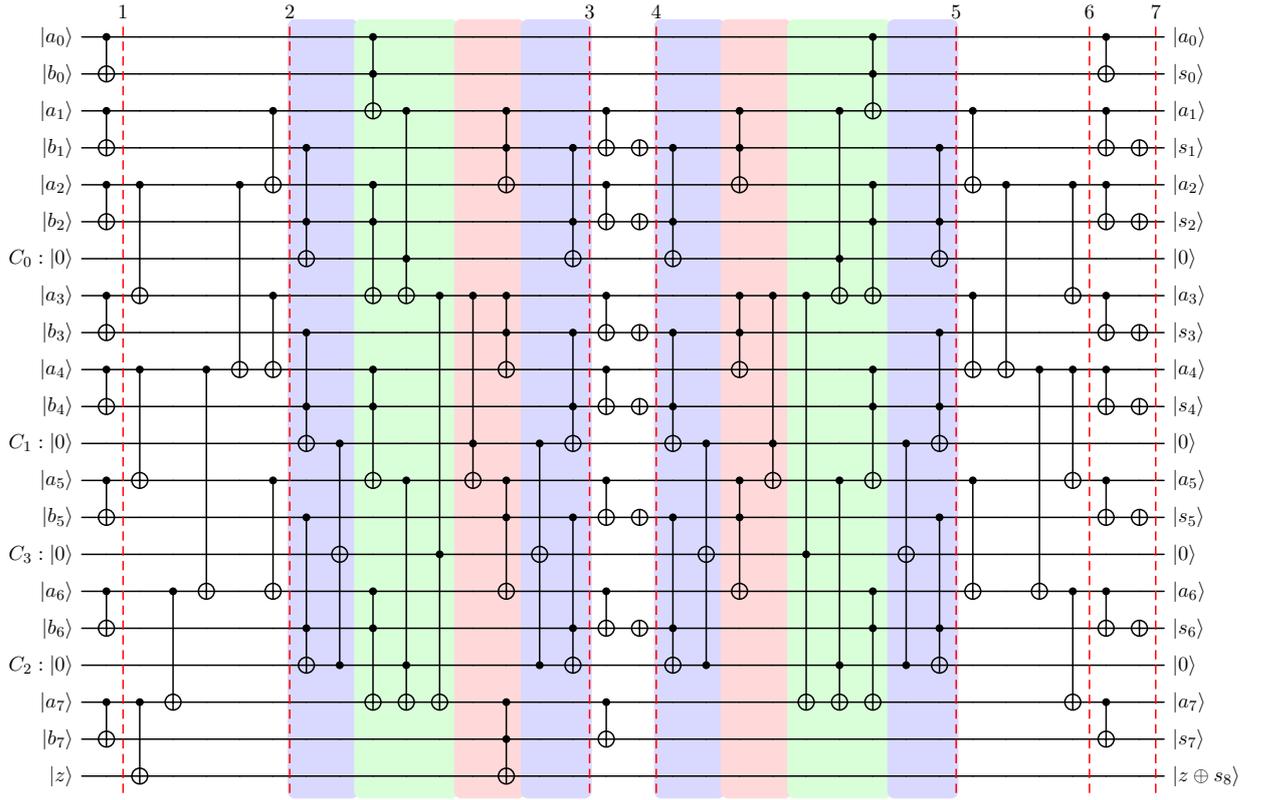
\begin{figure*}
    \centering
    \scalebox{0.7}{\input{images/Remaud}}
    \caption{Example of circuit generated by Algorithm~\ref{alg:main2} with $\ladder_2$ implemented in logarithmic depth (Lemma~\ref{lem:LadLog}) for $n=8$.}
    \label{fig:newadder}
\end{figure*}

%% file: images/Remaud.tex
\begin{quantikz}[column sep=.3cm, row sep={.7cm,between origins}]
	\lstick{$\ket{a_0}$}		& \ctrl{1} 	\slice{1} &	& 			& 			& 			& \slice{2}	& \gategroup[21,steps=2,style={draw=none,rounded corners,fill=blue!15, inner xsep=2pt},background]{} & 			& \ctrl{1} \gategroup[21,steps=3,style={draw=none,rounded corners,fill=green!15, inner xsep=2pt},background]{}  &			& & \gategroup[21,steps=2,style={draw=none,rounded corners,fill=red!15, inner xsep=2pt},background]{} 			&  &  \gategroup[21,steps=2,style={draw=none,rounded corners,fill=blue!15, inner xsep=2pt},background]{} & \slice{3}	&			& \slice{4}	& \gategroup[21,steps=2,style={draw=none,rounded corners,fill=blue!15, inner xsep=2pt},background]{} & 			& \gategroup[21,steps=2,style={draw=none,rounded corners,fill=red!15, inner xsep=2pt},background]{}	& 			& \gategroup[21,steps=3,style={draw=none,rounded corners,fill=green!15, inner xsep=2pt},background]{} &			&	\ctrl{1} & \gategroup[21,steps=2,style={draw=none,rounded corners,fill=blue!15, inner xsep=2pt},background]{} & \slice{5}	&			&			&			& \slice{6}	& \ctrl{1} 	& \slice{7}	& \rstick{$\ket{a_0}$} \\
	\lstick{$\ket{b_0}$}		& \targ{}	& 			& 			& 			& 			& 			&			&			& \ctrl{1}	&			&			&			&			&			&			&			&			&			&			&			&			&			&			& \ctrl{1}	&			&			&			&			&			&			& \targ{}	&			& \rstick{$\ket{s_0}$} \\
	\lstick{$\ket{a_1}$}		& \ctrl{1}	& 			& 			& 			& 			& \ctrl{2}	&			&			& \targ{}	& \ctrl{4}	&			&			& \ctrl{1}	&			&			& \ctrl{1}	&			&			&			& \ctrl{1}	&			&			& \ctrl{4}	& \targ{}	&			&			& \ctrl{2}	&			&			&			& \ctrl{1}	&			& \rstick{$\ket{a_1}$} \\
	\lstick{$\ket{b_1}$}		& \targ{}	& 			& 			& 			& 			& 			& \ctrl{2}	&			&			&			&			&			& \ctrl{1}	&			& \ctrl{2}	& \targ{}	& \targ{}	& \ctrl{2}	&			& \ctrl{1}	& 			&			&			&			&			& \ctrl{2}	&			&			&			&			& \targ{}	& \targ{}	& \rstick{$\ket{s_1}$} \\
	\lstick{$\ket{a_2}$}		& \ctrl{1}	& \ctrl{3}	& 			& 			& \ctrl{5}	& \targ{}	&			&			& \ctrl{1}	&			&			&			& \targ{}	&			&			& \ctrl{1}	&			&			&			& \targ{}	&			&			&			& \ctrl{1}	&			&			& \targ{}	& \ctrl{5}	&			& \ctrl{3}	& \ctrl{1}	&			& \rstick{$\ket{a_2}$} \\
	\lstick{$\ket{b_2}$}		& \targ{}	& 			& 			& 			& 			& 			& \ctrl{1}	&			& \ctrl{2}	&			&			&			&			&			& \ctrl{1}	& \targ{}	& \targ{}	& \ctrl{1}	&			&			&			&			&			& \ctrl{2}	&			& \ctrl{1}	&			&			&			&			& \targ{}	& \targ{}	& \rstick{$\ket{s_2}$} \\
	\lstick{$C_0: \ket{0}$}		&			& 			& 			& 			& 			& 			& \targ{}	&			&			& \ctrl{1}	&			&			&			&			& \targ{}	&			&			& \targ{}	&			&			&			&			& \ctrl{1}	&			&			& \targ{}	&			&			&			&			&			&			& \rstick{$\ket{0}$} \\
	\lstick{$\ket{a_3}$}		& \ctrl{1}	& \targ{}	& 			& 			& 			& \ctrl{2}	&			&			& \targ{}	& \targ{}	& \ctrl{7}	& \ctrl{4}	& \ctrl{1}	&			&			& \ctrl{1}	&			&			&			& \ctrl{1}	& \ctrl{4}	& \ctrl{7}	& \targ{}	& \targ{}	&			&			& \ctrl{2}	&			&			& \targ{}	& \ctrl{1}	&			& \rstick{$\ket{a_3}$} \\
	\lstick{$\ket{b_3}$}		& \targ{}	& 			& 			& 			& 			&			& \ctrl{2}	&			&			&			&			&			& \ctrl{1}	&			& \ctrl{2}	& \targ{}	& \targ{}	& \ctrl{2}	&			& \ctrl{1}	&			&			&			&			&			& \ctrl{2}	&			&			&			&			& \targ{}	& \targ{}	& \rstick{$\ket{s_3}$} \\
	\lstick{$\ket{a_4}$}		& \ctrl{1}	& \ctrl{3}	& 			& \ctrl{6}	& \targ{}	& \targ{}	&			&			& \ctrl{1}	&			&			&			& \targ{}	&			&			& \ctrl{1}	&			&			&			& \targ{}	&			&			&			& \ctrl{1}	&			&			& \targ{}	& \targ{}	& \ctrl{6}	& \ctrl{3}	& \ctrl{1}	&			& \rstick{$\ket{a_4}$} \\
	\lstick{$\ket{b_4}$}		& \targ{}	& 			& 			& 			& 			& 			& \ctrl{1}	&			& \ctrl{2}	&			&			&			&			&			& \ctrl{1}	& \targ{}	& \targ{}	& \ctrl{1}	&			&			&			&			&			& \ctrl{2}	&			& \ctrl{1}	&			&			&			&			& \targ{}	& \targ{}	& \rstick{$\ket{s_4}$} \\
	\lstick{$C_1: \ket{0}$}		&			& 			& 			& 			& 			& 			& \targ{}	& \ctrl{3}	&			&			&			& \ctrl{1}	&			& \ctrl{3}	& \targ{}	&			&			& \targ{}	& \ctrl{3}	&			& \ctrl{1}	&			&			&			& \ctrl{3}	& \targ{}	&			&			&			&			&			&			& \rstick{$\ket{0}$} \\
	\lstick{$\ket{a_5}$}		& \ctrl{1}	& \targ{}	& 			& 			& 			& \ctrl{3}	&			&			& \targ{}	& \ctrl{5}	&			& \targ{}	& \ctrl{1}	&			&			& \ctrl{1}	&			&			&			& \ctrl{1}	& \targ{}	&			& \ctrl{5}	& \targ{}	&			&			& \ctrl{3}	&			&			& \targ{}	& \ctrl{1}	&			& \rstick{$\ket{a_5}$} \\
	\lstick{$\ket{b_5}$}		& \targ{}	& 			& 			& 			& 			& 			& \ctrl{3}	&			&			&			&			&			& \ctrl{2}	&			& \ctrl{3}	& \targ{}	& \targ{}	& \ctrl{3}	&			& \ctrl{2}	&			&			&			&			&			& \ctrl{3}	&			&			&			&			& \targ{}	& \targ{}	& \rstick{$\ket{s_5}$} \\
	\lstick{$C_3: \ket{0}$}		&			& 			& 			& 			& 			& 			&			& \targ{}	&			&			& \ctrl{4}	&			&			& \targ{}	&			&			&			&			& \targ{}	&			&			& \ctrl{4}	&			&			& \targ{}	&			&			&			&			&			&			&			& \rstick{$\ket{0}$} \\
	\lstick{$\ket{a_6}$}		& \ctrl{1}	&			& \ctrl{3}	& \targ{}	& 			& \targ{}	&			&			& \ctrl{1}	&			&			&			& \targ{}	&			&			& \ctrl{1}	&			&			&			& \targ{}	&			&			&			& \ctrl{1}	&			&			& \targ{}	&			& \targ{}	& \ctrl{3}	& \ctrl{1}	&			& \rstick{$\ket{a_6}$} \\
	\lstick{$\ket{b_6}$}		& \targ{}	& 			& 			& 			& 			& 			& \ctrl{1}	&			& \ctrl{2}	&			&			&			&			&			& \ctrl{1}	& \targ{}	& \targ{}	& \ctrl{1}	&			&			&			&			&			& \ctrl{2}	&			& \ctrl{1}	&			&			&			&			& \targ{}	& \targ{}	& \rstick{$\ket{s_6}$} \\
	\lstick{$C_2: \ket{0}$}		&			& 			& 			& 			& 			& 			& \targ{}	& \ctrl{-3}	&			& \ctrl{1}	&			&			&			& \ctrl{-3}	& \targ{}	&			&			& \targ{}	& \ctrl{-3}	&			&			&			& \ctrl{1}	&			& \ctrl{-3}	& \targ{}	&			&			&			&			&			&			& \rstick{$\ket{0}$} \\
	\lstick{$\ket{a_7}$}		& \ctrl{1}	& \ctrl{2}	& \targ{}	& 			& 			& 			&			&			& \targ{}	& \targ{}	& \targ{}	&			& \ctrl{1}	&			&			& \ctrl{1}	&			&			&			&			&			& \targ{}	& \targ{}	& \targ{}	&			&			&			&			&			& \targ{}	& \ctrl{1}	&			& \rstick{$\ket{a_7}$} \\
	\lstick{$\ket{b_7}$}		& \targ{}	& 			& 			& 			& 			& 			&			&			&			&			&			&			& \ctrl{1}	&			& 			& \targ{}	&			&			&			&			&			&			&			&			&			&			&			&			&			&			& \targ{}	&			& \rstick{$\ket{s_7}$} \\
	\lstick{$\ket{z}$}			&			& \targ{}	& 			& 			& 			& 			&			&			&			&			&			&			& \targ{}	&			&			&			&			&			&			&			&			&			&			&			&			&			&			&			&			&			&			&			& \rstick{$\ket{z \oplus s_8}$}
\end{quantikz}

%% file: 4-conclusion.tex
\section{Conclusion}\label{sec:conclusion}

We have explicitly shown how several quantum adders based on the ripple-carry and carry-lookahead techniques can be linked. This link is the Toffoli ladder subroutine and optimizing its implementation effectively optimizes the implementation of quantum adders.

In addition, we have proposed a new carry-lookahead adder that has more interesting properties than its predecessors in the same family.

\subsection*{Acknowledgments}

The author would like to thank Vivien Vandaele for valuable discussions.

This work is part of HQI initiative (www.hqi.fr) and is supported by France 2030 under the French National Research Agency award number “ANR-22-PNCQ-0002”.